\documentclass{article}

\usepackage{epsfig}
\usepackage{amsmath}
\usepackage{bbm}
\usepackage{amsfonts}
\usepackage{amssymb}
\usepackage{amsthm}
\usepackage{graphicx}
\usepackage{comment}
\newtheorem{theorem}{Theorem}
\usepackage{cite}
\usepackage[T1]{fontenc}
\usepackage{authblk}

\usepackage{multirow}
\newcommand{\ket}[1]{\left\vert{#1}\right\rangle}
\newcommand{\bra}[1]{\left\langle{#1}\right\vert}
\newcommand{\Dpsi}[1]{D_{#1}\ket{\psi}}
\newcommand{\eye}{\mathbbm{1}}

\newtheorem{lemma}[theorem]{Lemma}
\theoremstyle{remark}

\newcommand{\bmt}{\begin{pmatrix}}
\newcommand{\emt}{\end{pmatrix}}
\newcommand{\vpu}[1]{^{\vphantom{#1}}}
\newcommand{\lsym}[2]{\genfrac{(}{)}{0.1 pt}{}{#1}{#2}}
\DeclareMathOperator{\SL}{SL}

\DeclareMathOperator{\LCM}{LCM}
\DeclareMathOperator{\Tr}{Tr}

\newcommand{\spacer}{\parbox{0.001 pt}{\rule{0 ex}{8 ex}}}

\title{Linear Dependencies in Weyl-Heisenberg Orbits}
\vspace{2cm}
\author[*,$\dag$]{{\normalsize Hoan Bui Dang}}
\author[$\ddag$]{{\normalsize Kate Blanchfield}}
\author[$\ddag$]{{\normalsize Ingemar Bengtsson}}
\author[*]{{\normalsize D. M. Appleby}}
\vspace{1cm}
\affil[*]{{\normalsize \textit{Perimeter Institute for Theoretical Physics, Waterloo, Ontario N2L 2Y5, Canada}}}
\affil[$\dag$]{{\normalsize \textit{Physics Department, University of Waterloo, Ontario N2L 3G1, Canada}}}
\affil[$\ddag$]{{\normalsize \textit{Stockholms universitet, AlbaNova, Fysikum, S-106 91 Stockholm, Sweden}}}

\date{}

\begin{document}

\maketitle

\begin{abstract}
Five years ago, Lane Hughston showed that some of the symmetric informationally complete positive operator valued measures (SICs) in dimension 3 coincide with the Hesse configuration (a structure well known to algebraic geometers, which arises from the torsion points of a certain elliptic curve). This connection with elliptic curves is signalled by the presence of linear dependencies among the SIC vectors. Here we look for analogous connections between SICs and algebraic geometry by performing computer searches for linear dependencies in higher dimensional SICs. We prove that linear dependencies will always emerge in Weyl-Heisenberg orbits when the fiducial vector lies in a certain subspace of an order 3 unitary matrix. This includes SICs when the dimension is divisible by 3 or equal to 8 mod 9. We examine the linear dependencies in dimension 6 in detail and show that smaller dimensional SICs are contained within this structure, potentially impacting the SIC existence problem. We extend our results to look for linear dependencies in orbits when the fiducial vector lies in an eigenspace of other elements of the Clifford group that are not order 3. Finally, we align our work with recent studies on representations of the Clifford group.
\end{abstract}

\newpage

\section{Introduction}

Symmetric informationally-complete positive operator-valued measures (SICs) \cite{Zau99, Ren04} represent a general form of measurement in quantum theory. As the more familiar projective measurements are associated to an orthonormal basis in Hilbert space, a SIC is associated to an over-complete set of $N^2$ unit vectors in $\mathcal{H}^{N}$ such that the absolute value of the scalar product between any distinct two is always constant, i.e.
\begin{equation} \sum_{\nu =0}^{N^2-1}\frac{1}{N}|\psi_{\nu} \rangle \langle \psi_{\nu}| = {\mathbbm 1}
\end{equation}
\begin{equation}
\left| \left\langle \psi_{\nu} | \psi_{\mu} \right\rangle \right| = \left\{ \begin{array}{ccc} 1 & \mbox{if} & \mbox{{\small $\nu = \mu$}} \\ \
\frac{1}{\sqrt{N+1}} & \mbox{if} &  \mbox{{\small $\nu \neq \mu$}} \end{array} \right.  
\end{equation}
SICs have practical applications in quantum state tomography \cite{Sco06, Zhu11}, quantum communication \cite{comm1,comm2}, quantum cryptography \cite{crypto1,crypto2}, classical high precision radar \cite{radar1,radar2} and classical speech recognition \cite{speech}. A significant amount of work is aimed towards proving their existence in all dimensions although no general proof is currently known. Numerical studies have successfully found SICs when $N \leq 67$ \cite{Sco10} and analytical solutions have been published for 20 of these dimensions (see references in \cite{Sco10} plus recent solutions for $N=16$ in \cite{Monomial} and $N=28$ in \cite{monomial2}). Such results promote the belief that SICs can always be found, but the current solutions are rather dimension-dependent and do not provide an overall coherent picture. 

All known SICs are group covariant, meaning they can be obtained from the action of a group on a single fiducial vector. The vast majority of SICs are covariant with respect to the Weyl-Heisenberg (WH) group and if the dimension is a prime, it has been shown this is the only possible group \cite{Zhu10}. In this paper we only consider such SICs. 

There are a few curious properties of SICs that are not yet understood, but clearly hint at an underlying mathematical structure. This scenario could support suggestions that SICs should play a deeper role in quantum mechanics \cite{Fuc09}. The most well-known of these properties is the conjecture of Zauner symmetry, which states that every SIC vector is an eigenvector of a certain order three unitary matrix \cite{Sco10,App05,Zau99}. Another property stems from current SIC solutions; in a preferred basis, the components of the SIC vectors are expressed as nested radicals. Such complex numbers are quite exceptional and arise because the Galois group of the polynomial equation one must solve is solvable \cite{Sco10,Hulya}. However, it is a third observation that we shall focus on in this paper, namely the connection between SICs and algebraic curves in dimension 3.

The WH group appears naturally in the theory of elliptic curves \cite{Mumford}, suggesting this might be an advantageous framework for the study of SICs. Furthermore, there is a very close connection between elliptic curves and linear dependencies among SIC vectors in dimension 3, first noted by Lane Hughston \cite{Hug}. This leads to the Hesse configuration in the complex projective plane---a particular structure well-known in algebraic geometry. We aim to explore this special connection, looking for analogous linear dependencies in higher dimensions. It was our hope that another configuration in complex projective space might reveal some new structure behind SICs and broaden the areas from which they can be studied. This did not turn out quite as planned for all the SICs we looked at but we do see similarities for certain SICs in dimension 8, which could signal another connection to elliptic curves. No sets of $N$ linearly dependent SIC vectors exist for $N=4$, 5 or 7, but we do find them again when $N=6$, 9 and 12---however, in these latter cases the linear dependencies are not uniquely tied to SICs. Additionally, we found an unexpected method of generating 2- and 3-dimensional SICs from the linear dependencies among vectors in 6- and 9-dimensional SICs, respectively. We don't know whether this can be generalised to all dimensions divisible by 3.

We observe that the reverse of our problem---to show that there exists an open set of fiducials giving rise to WH orbits without linear dependencies---has been considered for signal processing purposes and has been solved in prime dimensions \cite{harmonic1, harmonic2} and more recently in all finite dimension \cite{Romanos}.
 
In section 2 we review some basic facts about the Weyl-Heisenberg and Clifford groups, and explain notations and terminology used in the rest of the paper. Section 3 introduces the Hesse configuration and linear dependencies in SICs in dimension 3. Section 4 proves that linear dependencies can always be found in WH orbits when the fiducial vector is an eigenvector of an order 3 unitary matrix. In dimensions divisible by 3 and dimensions equal to 8 mod 9, these WH orbits include SICs. In section 5, we give numerical results from our computer search for dependencies in dimensions 4 to 9. We focus specifically on data relating to the SIC case in dimensions 6 and 9. In section 6, we prove the sets of linearly dependent vectors in dimension 6 contain 2-dimensional SICs and show that sets in dimension 9 contain 3-dimensional SICs. We do not have data for dimension 12 or higher, so do not know whether this pattern continues for all dimensions divisible by 3. In which case, this could open a new avenue for a SIC existence proof. Section 7 turns away from SICs altogether and proves that linear dependencies will always arise in WH orbits with a fiducial vector invariant under a unitary matrix whose order is a non-trivial divisor of $N$ in dimensions where $N$ is a product of distinct prime numbers. We give an example of this using linear dependencies in even dimensions. Section 8 relates our results to a recent representation of the Clifford group \cite{Monomial}. Section 9 summarises our results and points at some possible clues for the SIC existence problem.

\section{The Weyl-Heisenberg and Clifford groups}

In this section we review some basic facts about the Weyl-Heisenberg (WH) and Clifford groups which will be needed in the sequel. For more details see (for
example) \cite{App05}.  Let $|0\rangle, \dots, |N-1\rangle$ be the standard basis in dimension $N$. Define the operators $X$ and $Z$ by
\begin{align}
X |u\rangle & = |u+1\rangle \\ Z |u\rangle & = \omega^u |u\rangle 
\end{align}
where addition of ket-labels is mod $N$ and $\omega= e^{\frac{2 \pi i}{N}}$. The Weyl-Heisenberg displacement operators are then defined by
\begin{equation}
D_{\mathbf{p}} = \tau^{p_1 p_2} X^{p_1} Z^{p_2}
\end{equation}
where $\mathbf{p} = \left(\begin{smallmatrix} p_1
  \\ p_2 \end{smallmatrix}\right)$, $\tau= -e^{\frac{\pi i}{N}}$.
With this definition
\begin{align}
D_{\mathbf{p}}D_{\mathbf{q}} &= \tau^{\langle
  \mathbf{p},\mathbf{q}\rangle} D_{\mathbf{p}+\mathbf{q}} & D^{\dagger}_{\mathbf{p}} & = D^{\vphantom{\dagger}}_{-\mathbf{p}}
\label{eq:Dprd}
\end{align}
where the symplectic form $\langle \cdot , \cdot \rangle$ is defined by
\begin{equation}
\langle \mathbf{p}, \mathbf{q}\rangle = p_2q_1 -p_1 q_2
\end{equation}
It is convenient to introduce the notation
\begin{equation}
\bar{N} = \begin{cases} N \qquad & \text{$N$ odd} \\ 2N \qquad & \text{$N$ even} \end{cases}
\end{equation}
We then have $\tau^{\bar{N}} = 1$ and 
\begin{equation}
D_{\mathbf{p}} = D_{\mathbf{q}} \qquad \text{if $\mathbf{p} = \mathbf{q}$ mod $\bar{N}$}
\end{equation}
(note that if $N$ is even $\tau^N=-1$ and $\mathbf{p}=\mathbf{q}\mod N$ implies $D_{\mathbf{p}} = \pm D_{\mathbf{q}}$).
We  define the Weyl-Heisenberg group to consist of all operators of the form $\tau^n D_{\mathbf{p}}$, with $n \in \mathbb{Z}_{\bar{N}}$, $\mathbf{p} \in \mathbb{Z}_{\bar{N}}^2$ (where $\mathbb{Z}_{\bar{N}}$ is the set $\{0,1,\dots, \bar{N}-1\}$ equipped with arithmetic modulo $\bar{N}$).

In many situations the phase factor in the first of Equations (\ref{eq:Dprd}) is not important. It is accordingly convenient to introduce the notation $A \dot{=} B$ to signify that $A$ and $B$ are equal up to a phase.  We then have
\begin{equation}
D_{\mathbf{p}} D_{\mathbf{q}} \dot{=} D_{\mathbf{q}} D_{\mathbf{p}} \dot{=} D_{\mathbf{p}+\mathbf{q}}
\end{equation}
for all $\mathbf{p}$, $\mathbf{q}$.  Neglecting phases in this way the WH group reduces to the $N^2$ operators $\{D_{\mathbf{p}} \colon \mathbf{p} \in \mathbb{Z}_N^2\}$ (in more technical language these facts can be expressed by the statement that  the collineation group is $\cong \mathbb{Z}_N \times \mathbb{Z}_N$).  Similarly, if we ignore phases then the orbit of a vector $|\psi\rangle$ under the WH group reduces to the set $\{D_{\mathbf{p}}|\psi\rangle \colon \mathbf{p}\in\mathbb{Z}_N^2\}$ (in more technical language this set is the orbit of $|\psi\rangle$ considered as  an element of $\mathbbm{C}\mathrm{P}^{N-1}$ under the action of the collineation group).  In the sequel, by an abuse of terminology, we shall refer to this set  as the WH orbit of $|\psi\rangle$.

The symplectic group $\SL(2,\mathbb{Z}_{\bar{N}})$ consists of all
matrices
\begin{equation}
\mathcal{G} = \bmt \alpha & \beta \\ \gamma & \delta \emt
\end{equation}
such that $\alpha$, $\beta$, $\gamma$, $\delta \in
\mathbb{Z}_{\bar{N}}$ and $\det \mathcal{G} = 1$ (mod $\bar{N}$).  To each such
matrix there corresponds a unitary $U_\mathcal{G} $, unique up to a phase, such
that
\begin{equation}
U\vpu{\dagger}_{\mathcal{G}}  D\vpu{\dagger}_{\mathbf{p}} U^{\dagger}_{\mathcal{G} } =
D\vpu{\dagger}_{\mathcal{G} \mathbf{p}}
\end{equation}
If $\beta$ is relatively prime to $\bar{N}$ we have the explicit formula
\begin{equation}
U_{\mathcal{G} } = \frac{e^{i\theta}}{\sqrt{N}} \sum_{u,v=0}^{N-1}
\tau^{\beta^{-1}\left(\delta u^2 -2 u v+\alpha v^2\right)} |u\rangle
\langle v |
\label{eq:UFdef}
\end{equation}
where $\beta^{-1}$ is the multiplicative inverse of $\beta$ (mod
$\bar{N}$) and $e^{i\theta}$ is an arbitrary phase.    If
$\beta$ is not relatively prime to $\bar{N}$ we use the
fact~\cite{App05} that $\mathcal{G} $ has the decomposition
\begin{equation}
\mathcal{G} =\mathcal{G} _1\mathcal{G} _2 = \bmt \alpha_1 & \beta_1 \\ \gamma_1 & \delta_1 \emt \bmt
\alpha_2 & \beta_2 \\ \gamma_2 & \delta_2 \emt
\label{eq:primeDecomp}
\end{equation}
where $\beta_1$, $\beta_2$ both are relatively prime to $\bar{N}$ so
that $U_{\mathcal{G} _1}$, $U_{\mathcal{G} _2}$ can be calculated using
Equation~(\ref{eq:UFdef}).  $U_\mathcal{G} $ is then given by
\begin{equation}
U_{\mathcal{G} } = U_{\mathcal{G} _1} U_{\mathcal{G} _2}
\end{equation}
up to an arbitrary phase. We shall refer to unitaries of the form $e^{i \theta} U_{\mathcal{G}}$ as symplectic unitaries.  The Clifford group then consists of all
products $e^{i\theta}D_{\mathbf{p}} U_{\mathcal{G} }$. We have
\begin{align}
\left(D_{\mathbf{p}} U_{\mathcal{G} }\right)\left(D_{\mathbf{p}'} U_{\mathcal{G} '} \right)&\dot{=}D_{\mathbf{p}+\mathcal{G}\mathbf{p}'} U_{\mathcal{G}\mathcal{G}'}
&
\left(D_{\mathbf{p}} U_{\mathcal{G} }\right)^{\dagger} & \dot{=} D_{-\mathcal{G}^{-1}\mathbf{p}} U_{\mathcal{G}^{-1}}
\end{align}
for all $\mathbf{p},\mathbf{p}',\mathcal{G},\mathcal{G}'$.

As mentioned in the Introduction it turns out that  every known WH SIC fiducial vector is an eigenvector of a canonical order $3$ Clifford unitary and, conversely, that in every dimension where an exhaustive search has been made every canonical order $3$ Clifford  unitary has a SIC fiducial vector as one of its eigenvectors~\cite{App05,Sco10}.     In these statements the term ``canonical order $3$ Clifford unitary'' refers to a unitary of the form
\begin{align}
U =e^{i\theta} D_{\mathbf{p}} U_{\mathcal{G}}
\end{align}
where $\mathbf{p}$ is arbitrary and $\mathcal{G}$ is any element of $\SL(2,\mathbb{Z}_{\bar{N}})$ with the property
\begin{align}
\Tr(\mathcal{G}) = -1 \mod N
\end{align}
If $N=3$ one needs to impose the additional requirement that $\mathcal{G}\neq I$. It is then guaranteed that, with a suitable choice of the phase $e^{i\theta}$, $U$ is order $3$.   The matrix $\mathcal{G}$ is always conjugate~\cite{Flammia,Sco10}, either to the Zauner matrix 
\begin{equation}
\mathcal{Z} = \bmt 0 & -1\\ 1 & -1 \emt
\label{eq:canonical_Z}
\end{equation}
or its square or, in the case of dimensions $N = 9k+3$ with $k\ge 1$, to the matrix
\begin{equation}
\mathcal{A} = \bmt 1 & N+3 \\ N+3k & N-2 \emt
\end{equation}
In this paper we will confine ourselves to the case of fiducials which are eigenvectors of $U_{\mathcal{Z}}$.  Choosing the phase in Equation~(\ref{eq:UFdef}) to be $e^{i\theta} = e^{\frac{i\pi(N-1)}{12}}$ one finds~\cite{Zau99} that the eigenvalues of $U_{\mathcal{Z}}$ are $1,\eta,\eta^2$, where $\eta=e^{\frac{2\pi i}{3}}$.   We denote the three corresponding eigenspaces as $\mathcal{H}_1, \mathcal{H}_{\eta}$, and $\mathcal{H}_{\eta^2}$, and give their dimensions in Table 1.
\begin{table}[ht!]
\begin{center}
\begin{tabular}{c|c c c}
\hline
          & $N=3k$ & $N=3k+1$ & $N=3k+2$ \\ \hline
  $1$     & $k+1$  & $k+1$    & $k+1$ \\
  $\eta$     & $k$    & $k$      & $k+1$ \\
  $\eta^{2}$ & $k-1$  & $k$      & $k$ \\ \hline
\end{tabular}
\label{tab:subspaces_Z}
\caption{Multiplicities of the eigenvalues of  $U_{\mathcal{Z}}$ for different dimensions.}
\end{center}
\end{table}
Except in dimensions equal to $8 \mod 9$,  fiducials are only found in the highest dimensional eigenspace(s) of $U_{\mathcal{Z}}$.  When $N=8\mod 9$ fiducials are found in all 3 eigenspaces~\cite{Sco10}. We will hereafter refer to $\mathcal{H}_1$ (which, the numerical data suggests, always contains SIC fiducials) as \textit{the} Zauner subspace.

Finally, let us note that $\mathcal{Z}$ has non-trivial fixed points if and only if the dimension is divisible by $3$.  Specifically:
\begin{align}
\mathcal{Z}{\bf p} &= {\bf p} && \Leftrightarrow &{\bf p}& \in
\begin{cases}\spacer \left\{\bmt 0 \\ 0 \emt \right\} \qquad & \bar{N} \neq 0 \mod 3
\\
\spacer \left\{ \bmt 0 \\  0 \emt , \bmt \frac{\bar{N}}{3} \\ \frac{2\bar{N}}{3} \emt ,  \bmt \frac{2\bar{N}}{3} \\ \frac{\bar{N}}{3} \emt \right\} 
 \qquad & \bar{N} = 0 \mod 3
\end{cases} 
\label{eq:Z_fix_points}
\end{align}

\section{Linear dependencies in dimension 3}

Lane Hughston showed that vectors in some 3-dimensional SICs coincide exactly with the nine inflection points of a particular family of elliptic curves \cite{Hug}. The pattern of linear dependencies among the SIC vectors reproduces the Hesse configuration \cite{Hesse, Art09} and so we shall first look in more detail at linear dependencies coming from SICs in dimension 3. This section forms our motivation for studying higher dimensional dependency structures and, although SICs in dimension 3 have been studied fairly extensively before now \cite{Zhu10,App05}, we hope the following approach provides a useful way of looking at things.

In dimension 3, there is a continuous one-parameter family of SICs. The nine explicit (un-normalised) vectors are 
the columns of the matrix 
\begin{equation} 
\label{eq:sic_vectors}
\left[ \begin{array}{ccccccccc}
0 & 0 & 0 & -e^{i \theta} & -e^{i \theta} \eta & -e^{i \theta} \eta^2 & 1 & 1 & 1 \\
1 & 1 & 1 & 0 & 0 & 0 & -e^{i \theta} & -e^{i \theta} \eta & -e^{i \theta} \eta^2 \\
-e^{i \theta} & -e^{i \theta} \eta & -e^{i \theta} \eta^2 & 1 & 1 & 1 & 0 & 0 & 0 \\
\end{array} \right]
\end{equation}
The parameter $\theta$ needs only be considered in the interval $[ 0 , \frac{2 \pi}{6} ]$ as all other SICs are equivalent to these via transformations in the extended Clifford group. It is clear from inspection that any such SIC in dimension 3 will have three sets of three linearly dependent vectors. However, there are additional dependencies for certain choices of $\theta$. We can find these values by setting the determinant of any three other SIC vectors to zero. This gives
\begin{equation}
\left| \begin{array}{ccc}
0 & -e^{i \theta} \eta^m & 1 \\
1 & 0 & -e^{i \theta} \eta^n \\
-e^{i \theta} \eta^k & 1 & 0 \\
\end{array}
\right| = 1 - e^{3 i \theta} \eta^{k+m+n} = 0 
\end{equation}
and we find the condition
\begin{equation}
e^{3 i \theta} = \eta^r \quad \quad r \in \left\{ 0,1,2 \right\}  
\end{equation}
Given our limits for $\theta$, we are left with the options $\theta = 0$ or $\frac{2 \pi}{9}$ in order to obtain a `special' SIC with 12 sets of three linearly dependent vectors. The SIC with $\theta=0$ is a particularly special case as its fiducial vector
\begin{equation}
| \phi \rangle = \left( \begin{array}{c}
0 \\
1 \\
-1 \\
\end{array} \right)
\label{eq:sic_fid}
\end{equation}
is invariant under all symplectic unitaries. We can see this by rewriting the projector onto this vector as
\begin{equation}
| \phi \rangle \langle \phi | = \mathbbm{1} - U_{\mathcal{P}} \quad \quad \mathcal{P} = \left( \begin{array}{cc}
-1 & 0 \\
0 & -1 \\
\end{array} \right) 
\label{eq:sl_invariance}
\end{equation}
The matrix $\mathcal{P}$ is called the parity matrix, and its unitary representative $U_{\mathcal{P}}$ is one of the phase-point operators introduced by Wootters \cite{Woo87}. It is unique in being the only element (other than the identity element) invariant under conjugation by every other element in the symplectic group and so, with the identity element, forms the group's centre. As the right hand side of the first of Equations (\ref{eq:sl_invariance}) is invariant under $SL(2, {\mathbbm Z}_N)$, the SIC fiducial vector $| \phi \rangle$ must be an eigenvector of every symplectic unitary matrix. Of interest to us is its invariance under the action of the symplectic canonical order 3 unitaries. There are eight such unitaries altogether (including the Zauner unitary) but we consider the four unitaries
\begin{eqnarray}
U_{ \left( \begin{smallmatrix} 1&0\\ 1&1 \end{smallmatrix} \right)} = \left( \begin{array}{ccc} \eta & 0 & 0 \\ 0 & 1 & 0 \\ 0 & 0 & 1 \end{array} \right) \quad 
U_{\left( \begin{smallmatrix} 0&2\\ 1&2 \end{smallmatrix} \right)} = \frac{e^{\frac{i \pi}{6}}}{\sqrt{3}} \left( \begin{array}{ccc} 1 & 1 & 1 \\ \eta^2 & 1 & \eta \\ \eta^2 & \eta & 1 \end{array} \right) \nonumber \\ \\
U_{\left( \begin{smallmatrix} 2&2\\ 1&0 \end{smallmatrix} \right)} = \frac{e^{\frac{i \pi}{6}}}{\sqrt{3}} \left( \begin{array}{ccc} 1 & \eta^2 & \eta^2 \\ 1 & 1 & \eta \\ 1 & \eta & 1 \end{array} \right) \quad 
U_{\left( \begin{smallmatrix} 1&2\\ 0&1 \end{smallmatrix} \right)} = \frac{e^{\frac{i \pi}{6}}}{\sqrt{3}} \left( \begin{array}{ccc} 1 & \eta & \eta \\ \eta & 1 & \eta \\ \eta & \eta & 1 \end{array} \right) \nonumber
\label{eq:sco3u}
\end{eqnarray}
since the other four are squares of these. Note that $U_{\left( \begin{smallmatrix} 0&2\\ 1&2 \end{smallmatrix} \right)}$ corresponds to the Zauner matrix in Equation (\ref{eq:canonical_Z}) of Section 2. The SIC fiducial defined in Equation (\ref{eq:sic_fid}) is then precisely the intersection point of the 2-dimensional subspaces $\mathcal{H}_1$ from each of these four symplectic canonical order 3 unitaries. We can then ask whether any other SIC vectors are also contained within these same subspaces. It turns out that the remaining SIC vectors are distributed evenly amongst the subspaces: 2 further vectors per unitary subspace. 

When we view the full SIC as a $3 \times 3$ lattice in the complex projective plane (where each point is a SIC vector and each line is a 2-dimensional subspace) we generate the Hesse configuration. In Figure \ref{fig:hesse}, the fiducial vector is found at the point at the bottom left of the plane. The left-hand image shows the three linear dependencies in every dimension 3 SIC while the right-hand image shows the pattern of symplectic canonical order 3 unitary subspaces that intersect at the fiducial vector $| \phi \rangle$, i.e. the SIC fiducial for which $\theta=0$. The Clifford group also contains WH conjugates of the form $D_{p} U D_{p}^{\dagger}$, where $U$ is one of the four unitaries in Equation (27), and in this way we find a total of 12 order three unitaries and hence 12 subspaces. Remarkably, they always intersect in sets of four at a point described precisely by one of the SIC vectors. We recover the famous Hesse configuration of 9 points and 12 lines in the projective plane \cite{Hesse}. To obtain this full picture from Figure \ref{fig:hesse}, we must translate the lines in the right-hand image. Each line can be translated three times, so we find 12 lines in total.

\begin{figure}[ht!]
\centering
\includegraphics[width=0.9\textwidth]{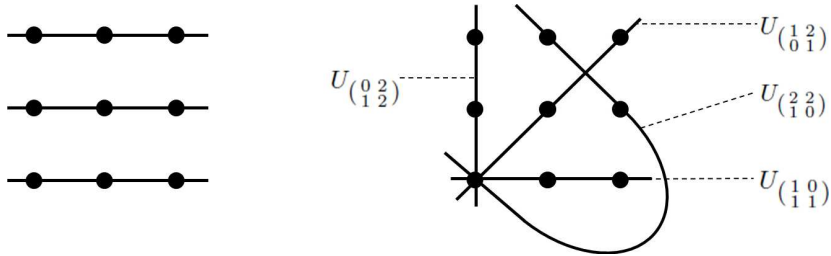}
\caption{A SIC in the complex projective plane. The left-hand image shows the three sets of three linear dependencies in every SIC and the right-hand image shows the four subspaces from Equation (27) that intersect at the fiducial vector with $\theta=0$.}
\label{fig:hesse}
\end{figure}

Hesse originally discovered his configuration by examining elliptic curves. In particular, the family of cubics 
\begin{equation}
x^3 + y^3 + z^3 + \lambda x y z = 0 \quad \quad \lambda \in \mathbbm{C}
\label{eq:hesse_cubic}
\end{equation}
is invariant under the WH group. It is often called the Hesse pencil of cubics. The points on the curve where the determinant of the Hessian (the matrix containing the second derivatives of the curve) vanish are called inflection points. The determinant of the Hessian is also a cubic in the same family and B\'{e}zout's theorem states that two cubics in $\mathbbm{C}\mathrm{P}^{2}$ will intersect in nine points. These are the nine inflection points and they coincide exactly with the nine SIC vectors, i.e. then the values of $x$, $y$ and $z$ in Equation (28) are the components in a SIC vector. The same inflection points (and thus the same SIC) arise for all values of $\lambda$, except for the special cases when $\lambda = \infty$ and $\lambda^3 = -1$. These special cases give rise to singular members of this family. Each singular curve, $T_i$, degenerates into three projective lines
\begin{equation}
\begin{array}{l}
T_{0} : xyz = 0 \\
T_{1} : (x+y+z)(x+\eta y+\eta^{2}z)(x+\eta^{2}y+\eta z) = 0 \\
T_{2} : (x+\eta y+z)(x+\eta^{2}y+\eta^{2}z)(x+y+\eta z) = 0 \\
T_{3} : (x+\eta^{2}y+z)(x+\eta y+\eta z)(x+y+\eta^{2}z) = 0 \\
\end{array} 
\end{equation}
Each line passes through three inflection points and each inflection point lies on four lines, forming our familiar Hesse configuration. It is often denoted $\left( 9_4 , 12_3 \right)$ to reflect the property that the 9 points each lie in 4 lines and the 12 lines each contain 3 points. A more detailed survey of the Hesse configuration is given in \cite{Art09} and its connection to SICs in \cite{Ingemar}.

This pattern of subspace intersections has a connection to mutually unbiased bases. Each of the twelve lines in the Hesse configuration represents a 2-dimensional subspace in $\mathbbm{C}^3$ with an orthogonal vector that is uniquely determined up to a scalar. We shall call this a normal vector from here onwards. These 12 normal vectors reproduce the standard MUB in dimension 3, obtained from the eigenvectors of the cyclic subgroups of the WH group. Note that we use the acronym MUB to mean a complete set of $N+1$ mutually unbiased bases. As the MUB vectors lie orthogonally to the 2-dimensional eigenspaces, they must sit in the 1-dimensional eigenspace $\mathcal{H}_{\eta}$.

So far we have been looking at the single SIC with $\theta=0$. However, one also gets a Hesse configuration when $\theta = \frac{2 \pi}{9}$. Most SICs in dimension 3 are part of an extended Clifford group orbit containing eight SICs \cite{App05}. The exceptions are the SICs with $\theta=0$ and $\theta=\frac{2 \pi}{6}$ in Equation (\ref{eq:sic_vectors}), which lie on orbits containing one and four SICs, respectively. The `special SIC', with $\theta=\frac{2 \pi}{9}$, is then on an extended Clifford orbit with seven other SICs and all eight of these SICs are special in the sense that they have 12 linearly dependent sets of three vectors. This extends to the MUB connection. Above, we found a MUB from the normal vectors---lines in the Hesse configuration---of linear dependencies in the `special SIC' with $\theta=0$. Note that this MUB is itself a single orbit of the extended Clifford group. We can extract eight more MUBs in the same way from the other eight `special SICs'. They can be arrived at in a different way. The Clifford group contains altogether 36 order 3 subgroups with elements of the form $D_{\bf p} U$, with $U$ a symplectic order 3 unitary. So far we have used only 12 of these. The remaining 24 have non-degenerate spectra, and their eigenbases can be collected into 8 sets of 3 mutually unbiased bases. Each such triplet is in itself an orbit under the WH group, and forms a MUB when a basis from the standard MUB is added \cite{Alltop,Kate}. This ties in exactly with the eight `special SICs' with $\theta=\frac{2 \pi}{9}$ through the Hesse configuration \cite{Kate}. Each SIC, together with one of these WH orbit MUBs, reproduces the pattern $\left( 9_4 , 12_3 \right)$. It is perhaps interesting to note that we could have arrived at the nine `special SICs' from knowledge of the symplectic canonical order 3 unitaries and the additional WH orbit MUBs alone.

\section{Linear dependencies from the eigenvectors of $U_{\mathcal{Z}}$}

In this section we show how linear dependencies will arise in every dimension. The results in this section apply to all orbits in which the fiducial vector lies in certain eigenspaces of $U_{\mathcal{Z}}$, depending on the dimension, and not only to SIC fiducials.

From Section 2, we can label a vector in a WH orbit, up to a phase, by a 2-component vector ${\bf p}$ in $\mathbb{Z}_N^2$. In this representation, we are interested in the orbit of ${\bf p}$ under the action of $\mathcal{Z}$. If ${\bf p}$ is one of the fixed points of $\mathcal{Z}$ it will form a singlet, otherwise $\{{\bf p},\mathcal{Z}{\bf p},\mathcal{Z}^2{\bf p}\}$ forms a triplet since  $\mathcal{Z}$ is of order 3. If the dimension $N$ is divisible by 3, $\mathcal{Z}$ has exactly 3 singlets, which are given by Equation (\ref{eq:Z_fix_points}). In other dimensions, the only singlet is  ${\bf p}=(0,0)$. In the Hilbert space, we will use the same terminology to call $D_{\bf{p}}\ket{\psi}$ a singlet if ${\bf p}$ is a singlet, and to call $\{D_{\bf p}\ket{\psi},D_{\mathcal{Z}{\bf p}}\ket{\psi},D_{\mathcal{Z}^2{\bf p}}\ket{\psi}\}$ a triplet otherwise.

\begin{theorem}
\label{theoremLD1}
In dimension $N=3k$ any subset of $N$ vectors in a WH orbit whose fiducial vector is an eigenvector of $U_{\mathcal{Z}}$ is linearly dependent if it contains $k$ triplets, or $k-1$ triplets and 3 singlets.
\end{theorem}

\begin{proof}
Let $N=3k$ and let $U_{\mathcal{Z}}\ket{\psi_0}=\lambda \ket{\psi_0}$. When ${\bf p}$ is a singlet, $D_{\bf p}\ket{\psi_0}$ lies in the same eigenspace of $U_{\mathcal{Z}}$ as $\ket{\psi_0}$ because
\begin{equation}
U_{\mathcal{Z}} D_{\bf p} \ket{\psi_0} = U_{\mathcal{Z}} D_{\bf p} U_{\mathcal{Z}}^{\dagger} U_{\mathcal{Z}}\ket{\psi_0} = \lambda D_{\mathcal{Z} {\bf p}}\ket{\psi_0} =\lambda D_{\bf p}\ket{\psi_0} 
\label{eq:Z_transformation}
\end{equation}
When ${\bf p}$ is in a triplet, we construct three new linear combinations of vectors in the triplet $\{D_{\bf p}\ket{\psi},D_{\mathcal{Z}{\bf p}}\ket{\psi},D_{\mathcal{Z}^2{\bf p}}\ket{\psi}\}$. Our vectors are
\begin{equation}
\left| r \right\rangle = D_{\bf{p}} \left| \psi_{0} \right\rangle + U_ {\mathcal{Z}} D_{\bf{p}} \left| \psi_{0} \right\rangle + U_ {\mathcal{Z}}^{2} D_{\bf{p}} \left| \psi_{0} \right\rangle 
\nonumber\end{equation}
\begin{equation}
\left| s \right\rangle = D_{\bf{p}} \left| \psi_{0} \right\rangle + \eta^2 U_ {\mathcal{Z}} D_{\bf{p}} \left| \psi_{0} \right\rangle + \eta U_ {\mathcal{Z}}^{2} D_{\bf{p}} \left| \psi_{0} \right\rangle 
\end{equation}
\begin{equation}
\left| t \right\rangle = D_{\bf{p}} \left| \psi_{0} \right\rangle + \eta U_ {\mathcal{Z}} D_{\bf{p}} \left| \psi_{0} \right\rangle + \eta^2 U_ {\mathcal{Z}}^{2} D_{\bf{p}} \left| \psi_{0} \right\rangle \nonumber
\end{equation}
where $\eta=e^{\frac{2 \pi i}{3}}$. We will refer to vectors constructed in this way, for a given choice of {\bf p}, as $r$-type, $s$-type and $t$-type respectively. It is straightforward to check that these vectors are evenly distributed among the three eigenspaces $\mathcal{H}_{1}$, $\mathcal{H}_{\eta}$ and $\mathcal{H}_{\eta^2}$, shown in Table 1. It is also clear that the linear span of the vectors $\{ D_{\bf{p}} \ket{\psi_0}, D_{\mathcal{Z} \bf{p}} \ket{\psi_0}, D_{\mathcal{Z}^{2} \bf{p}} \ket{\psi_0}\}$  equals that of the vectors $\{\ket{r} , \ket{s} , \ket{t}\}$.

If a set of $N$ vectors contains $k$ triplets, this gives $k$-many of each $r$-,$s$-, and $t$-type vector. From Table 1 we know that the $r$-type vectors lie in an eigenspace of dimension $k+1$, the $s$-type vectors lie in an eigenspace of dimension $k$, and the $t$-type vectors lie in an eigenspace of dimension $k-1$. It is clear then that the $k$ $r$-type vectors cannot fully span their subspace while the $k$ $t$-type vectors are overcomplete and therefore linearly dependent.

If a set contains $k-1$ triplets and $3$ singlets, this gives $(k-1)$-many of each $r$-,$s$-, and $t$-type vector, plus the 3 singlets that are in the same eigenspace as the fiducial vector $\ket{\psi_0}$. Therefore there will be $k+2$ vectors lying in the same eigenspace. Since the largest eigenspace of $U_{\mathcal{Z}}$ has dimensionality $k+1$, this results in linear dependency. 
\end{proof}

When the fiducial vector lies in $\mathcal{H}_\eta$ or $\mathcal{H}_{\eta^2}$ we can obtain dependencies using fewer than $N$ vectors. For example, if $\ket{\psi_0} \in \mathcal{H}_{\eta}$ then linear dependency will also occur when the set contains $k-1$  triplets and $2$ singlets, or $k-2$ triplets and $3$ singlets. And if $\ket{\psi_0} \in \mathcal{H}_{\eta^2}$, linear dependency will also occur when the set contains $k-1$ triplets and $1$ singlet, $k-2$ triplets and $2$ singlets, or $k-3$ triplets and $3$ singlets.

If the dimension is not divisible by 3 the matrix $\mathcal{Z}$ has only one fixed point. Nevertheless WH orbits with linear dependencies do arise.

\begin{theorem}
In dimension $N=3k+1$ any subset of $N$ vectors in a WH orbit whose fiducial vector lies in the eigenspace $\mathcal{H}_{\eta}$ or $\mathcal{H}_{\eta^2}$ is linearly dependent if it contains $k$ triplets and 1 singlet.
\end{theorem}

\begin{proof}
In dimension $N=3k+1$, the three eigenspaces $\mathcal{H}_{1}$, $\mathcal{H}_{\eta}$ and $\mathcal{H}_{\eta^2}$ have dimensionality $k+1$, $k$, and $k$, respectively. If the fiducial lies in $\mathcal{H}_{\eta}$, which is the only singlet in this case, then together with $k$ $s$-type vectors it will give $k+1$ vectors in this subspace, resulting in linear dependency. The same argument applies when the fiducial vector lies in $\mathcal{H}_{\eta^2}$.
\end{proof}

\begin{theorem}
In dimension $N=3k+2$ any subset of $N-1$ vectors in a WH orbit whose fiducial vector lies in the eigenspace $\mathcal{H}_{\eta^2}$ is linearly dependent if it contains $k$ triplets and 1 singlet.
\end{theorem}

\begin{proof}
In dimension $N=3k+2$, the eigenspace $\mathcal{H}_{\eta^2}$ has dimensionality $k$. The $k$ $t$-type vectors, together with the fiducial vector, form a set of $k+1$ vectors in $\mathcal{H}_{\eta^2}$, thus they are linearly dependent.
\end{proof}

Lastly, we note that if a set is linearly dependent, all other sets in its WH orbit are also linearly dependent. In Theorem \ref{theoremLD1}, sets that contain $k$ triplets, or $k-1$ triplets and $3$ singlets, have the property that they remain invariant (as sets) under the action of  $U_{\mathcal{Z}}$. We shall refer to these sets as Zauner invariant sets, and their orbits as Zauner invariant orbits.

\section{Numerical linear dependencies}

In the previous section we proved that we can find $N$ linearly dependent vectors in a WH orbit when the fiducial vector is in a certain Zauner subspace. However, this doesn't quite capture the whole picture. 
In this section, we provide the outcome of our numerical search for linear dependencies in dimensions 4 to 8 (with partial results when $N=9$ and 12) and give additional linear dependencies that are not covered in the previous section. We are especially interested in dimensions divisible by 3, where dependencies are always found if the fiducial lies in the subspace where we expect the SIC fiducials to be, and we discuss details of linear dependencies in dimensions 6 and 9. The additional dependencies in these dimensions do not depend on whether the WH orbit is a SIC or not, although we find some interesting orthogonality relationships in the dependency structure that are unique to SICs (and hence do not occur in WH orbits starting with a non-SIC fiducial vector in $\mathcal{H}_{1}$). Our investigation into dependencies in dimensions 6 and 9 also led to smaller-dimensional SICs, but we will postpone a discussion of these until the next section. In dimension 8, there are SIC fiducials in the Zauner subspace $\mathcal{H}_{\eta^2}$ and our numerical data show that these SICs have more linearly dependent sets of 8 vectors than WH orbits with a vector in $\mathcal{H}_{\eta^2}$. Recall that this was precisely the situation with the `special SICs' in dimension 3 and so could indicate a dimension 8 analogue of the Hesse configuration.

In each dimension $N$, our computer program started with a vector from one of the Zauner subspaces, generated the full orbit under the action of the WH group and then performed an exhaustive search for all sets of $N$ vectors that are linearly dependent. For each subspace of the Zauner unitary in each dimension, we repeated the procedure with a small number of arbitrarily chosen fiducial vectors to ensure that the results are the same for generic fiducials. In the case of SIC vectors, we used the fiducials given in \cite{Sco10}; where there was a choice of Clifford orbits, we repeated the calculation with fiducials from each orbit.

We found no dependence on the choice of fiducial except in dimension 8 where the SIC gives 24,935,160 linear dependencies, which is slightly higher than the generic result given in Table 2. The results often showed a higher number of linear dependencies than accounted for by Theorems 1-3. Table 2 shows the total number of sets of $N$ linearly dependent vectors in WH orbits with fiducial vectors in different Zauner subspaces. The number predicted by the proofs in the previous section, or an upper bound on this number, is given in brackets. For dimension 9 we were not able to perform an exhaustive search.

\begin{table}[ht!]
\begin{center}
\begin{tabular}{c | c c c c c}
\hline
N & 4 & 5 & 6 & 7 & 8\\
\hline 
$\mathcal{H}_1$ & 0 & 0 & 984 & 0 & 0\\
 & 0 & 0 & (768) & 0 & 0 \\
$\mathcal{H}_{\eta}$ & 116 & 0 & 635052 & 5796 & 0\\
 & (68) & 0 & (75342) & (5796) & 0 \\
$\mathcal{H}_{\eta^2}$ & 116 & 6600 & 1790328 & 5796 & 24756984 \\
 & (68) & (4200) & - & (5796) & ($\le$766080) \\
\hline
\end{tabular}
\label{tab:dependencies}
\caption{Linear dependencies in a WH orbit when the initial vector is taken from each Zauner subspace. The numbers in brackets are the total number of sets (or in one case an upper bound on this number) predicted in Section 4.}
\end{center}
\end{table}

Assuming the Zauner conjecture holds, Theorem 1 shows that SICs in dimensions divisible by 3 contain linearly dependent vectors. 
As we are primarily interested in SICs, we turn to look in more detail at these dimensions. When $N=6$ there are 984 sets of six linearly dependent SIC vectors, to a numerical precision of $10^{-15}$. All the results reported in this section are numerical and hold to within this precision. This is a higher number of dependencies than we proved must exist in the previous section, which finds only 768 of the 984 sets we generated numerically. Recall that the proof required the N vectors in the WH orbit to 
be invariant under the Zauner unitary or one of its WH conjugates; the additional 216 sets in dimension 6 are not of this form. Curiously, they are instead invariant under an order 6 unitary matrix $U_\mathcal{M}$ whose symplectic representation is
\begin{equation}
\mathcal{M} = \left( \begin{array}{ll} 
3 & 8 \\ 
4 & 11 \end{array} \right) 
\end{equation}

Each of the 36 SIC vectors lies in 164 different sets of linearly dependent SIC vectors, and clearly each of the 984 linearly dependent sets contains 6 vectors. In the language of complex projective space there are 36 points and 984 hyperplanes, forming the balanced configuration $(36_{164}, 984_6)$. If this constitutes a known pattern in $\mathbbm{C}\mathrm{P}^{5}$ we do not recognise it. 

The 984 linearly dependent sets can themselves be collected into orbits under the WH group. There are 27 orbits of length 36 and 1 orbit of length 12. The short orbit arises because it contains only linearly dependent sets invariant under the subgroup $\{1,X^2Z^4,X^4Z^2\}$. This subgroup commutes with the Zauner unitary defined in Equations (\ref{eq:UFdef}) and (\ref{eq:canonical_Z}), which leaves the fiducial SIC vector invariant. Twenty two of these WH orbits contain sets that are invariant under the Zauner matrix (or sets invariant under a WH conjugate $D_{{\bf p}} U_{\mathcal{Z}} D_{{\bf p}}^{\dagger}$ of the Zauner matrix); the other 6 are invariant under the action of $U_\mathcal{M}$.

Inspired by the 3-dimensional case, we look at the 984 vectors normal to the linearly dependent sets and perform an exhaustive search for orthogonalities between them. There is no basis among them---nor a MUB!---but there are nine sets of four mutually orthogonal normal vectors. These normal vectors all belong to the same orbit under the WH group (of length 36). Each set of four mutually orthogonal vectors forms an orbit under the subgroup of the WH group that contains order 2 elements only. In addition, there are over 20,000 orthogonal triples of normal vectors, only 712 of which lie in a single WH orbit.

The resulting structure in dimension 6 is summarised in Table 3, where we have labelled each WH orbit from 1 to 28. The first orbit is the short one and we see that every set of six linearly dependent vectors in this orbit is formed from two orbits under the subgroup $\{1,X^2Z^4,X^4Z^2\}$. The orbit containing the mutually orthogonal normal vectors is number 11, and orbits 23 to 28 are the ones invariant under $U_\mathcal{M}$---they contain sets of six vectors forming just one orbit under $U_\mathcal{Z}$, rather than two.

\begin{table}[ht!]
\begin{center}
\begin{tabular}{ c | c c c }
\hline
WH orbit & No. orbits under  & No. orbits under                & No. ON \\
         & $U_{\mathcal{Z}}$ & $\{ {\bf 1}, X^2Z^4, X^4Z^2\}$  & quadruples \\
\hline
1 & 2 & 2 & 0 \\ 
2--10 & 2 & 1 & 0 \\ 
11 & 2 & 0 & 9 \\ 
12--13 & 2 & 0 & 0 \\ 
14--22 & 2 & 0 & 0 \\ 
23--28 & 1 & 1 & 0 \\ 
\hline
\end{tabular}
\label{tab:WH_orbits}
\caption{Properties of WH orbits starting with a vector in $\mathcal{H}_{1}$ in dimension 6.}
\end{center}
\end{table}

The six-dimensional case shows some similarities to the three-dimensional one, but the pattern is not as striking. We now ask to what extent the results depend on the fact that we start from 36 vectors forming a SIC. From the previous section, we expect the pattern of at least 768 of the 984 sets of linearly dependent vectors to remain unchanged, i.e. the same {\bf p} vectors to appear in the same sets. To check this, we repeated our calculations starting from a fiducial vector chosen arbitrarily from the Zauner subspace (i.e. in $\mathcal{H}_1$). This will produce 36 vectors typically not forming a SIC. In the cases we looked at, we found all 984 sets of six linearly dependent vectors had exactly the same pattern as when we started with a SIC.

However, there are some subtle differences between SICs and non-SICs in dimension 6, despite the pattern of dependencies being the same when the fiducial is taken arbitrarily from $\mathcal{H}_1$. The four orthogonal quadruples among the normal vectors are present only in the SIC case. An additional 216 orthogonal triples vanish when we pick an arbitrary fiducial from $\mathcal{H}_1$; the triples all have one normal vector in WH orbit 11, one vector from either WH orbit 12 or 13 and the final vector from one of the orbits 23, 24,$\ldots$,28. As every normal vector in orbits 23 to 28 are used, we find $6 \times 36 = 216$ triples. It is perhaps noteworthy that these latter six `SIC-sensitive' WH orbits are precisely the $U_\mathcal{M}$ invariant ones.

In dimension 9, the SIC has 79,767 sets of nine linearly dependent vectors. Of these, 78,795 are predicted by Theorem 1 in Section 4. This large number of dependencies meant that we could not perform as complete computer searches as for the dimension 6 case. For example, we did not compute the scalar products between every pair of normal vectors, although we happen to know they form a basis (this is discussed in more detail in section 8). As with the dimension 6 case, we can collect these 79,767 linearly dependent sets into orbits under the WH group. We find 987 WH orbits: 984 of length 81, 2 of length 27 and 1 of length 9. This gives 975 orbits invariant under the action of the Zauner unitary or one of its WH conjugates, with the remaining 12 invariant under $U_\mathcal{M}$ (186 orbits are invariant under both types). In odd dimensions where $N=3k$, $\mathcal{M}$ takes the form
\begin{equation}
G = \left( \begin{array}{ll} 
k+1 & k \\ 
2k & 2k+1 \end{array} \right) 
\end{equation}

When we repeat the calculations using an arbitrary vector from $\mathcal{H}_1$, we again find exactly the same pattern of linearly dependent vectors (in the sense that the dependencies are labelled by the same {\bf p} vectors). This suggests a distinction between SICs in dimension 3 and SICs in higher dimensions divisible by 3. `Special SICs' in dimension 3 gave rise to 12 sets of linear dependencies, while the others produced only 3. No SICs are `special' in this sense in dimensions 6 and 9; instead we find that not only do all SICs have the same pattern of linearly dependent vectors, but this pattern can be be almost completely reproduced by acting with the WH group on any vector in the Zauner subspace (the pattern of orthogonalities among normal vectors cannot be reproduced). However, we did find one other instance of a SIC fiducial vector giving more linear dependencies than other vectors in the same subspace. In dimension 8, there is a SIC sitting in another eigenspace of the Zauner unitary, namely the one with the smallest dimension, $\mathcal{H}_{\eta^2}$. This SIC has more linear dependencies than a WH orbit starting with an arbitrary vector from $\mathcal{H}_{\eta^2}$: the SIC exhibits 24,935,160 sets of eight linearly dependent vectors while a WH orbit using a vector in $\mathcal{H}_{\eta^2}$ exhibits 24,756,984 (it is this latter value that is shown in Table 2) \cite{Gra_pc}. This may be connected to the fact that the SIC has a larger automorphism group than an arbitrarily chosen vector. We notice that the Hessian SIC in dimension 3 has more linear dependencies that an arbitrarily chosen vector in the same eigenspace. As SICs can always be found in this smaller subspace when the dimension equals 8 mod 9, one might speculate whether a similar connection to elliptic curves exists in this family of dimensions.

We were unable to collect exhaustive information in dimension 12 regarding linear dependencies, so do not know how many sets of 12 linearly dependent vectors are in a SIC or a WH orbit starting from an arbitrary vector in $\mathcal{H}_1$ (or indeed any Zauner subspace). Similarly, we do not know whether their normal vectors exhibit any orthogonality structure analogous to the dimension 6 or 9 cases.

\section{Small SICs in dimension $N=6$}

In dimension $N=6$, among 984 vectors normal to the linear dependent sets generated from an arbitrary fiducial vector in the Zauner subspace (not necessarily a SIC fiducial), we found 30 sets of 4 normal vectors that form 2-dimensional SICs, i.e. in each set, the absolute values for the overlaps between the vectors are $1/\sqrt{3}$ and they lie on a 2-dimensional subspace. This phenomenon also happens in dimension $N=9$, where 3-dimensional SICs are found among the normal vectors. We refer to SICs of this kind as small SICs (because their dimension is smaller than that of the whole space). Attempts to find such small SICs in dimension $N=12$ yielded no positive result, perhaps only due to the impracticality of an exhaustive search. In this section, we give an explanation for the small SICs in dimension $N=6$. In dimension $N=9$, the phenomenon is not yet fully understood.

The 2-dimensional small SIC sets found in dimension $N=6$ are all of the form $\{\ket{\psi}, D_{03} \ket{\psi}, D_{30} \ket{\psi}, D_{33} \ket{\psi}\}$, where $D_{ij}$'s are the displacement operators defined in Section 2. As will be proven in the following theorem, for such sets to form a 2-dimensional SIC, $\ket{\psi}$ just needs to be any normalised eigenvector of $U_\mathcal{Z}$ that is not in the Zauner subspace.  Normal vectors of linearly dependent sets that contain three singlets in the case when the fiducial lies in the Zauner subspace just happen to satisfy this condition, as can been seen from Section 4.

\begin{theorem}
In dimension $N=6$, if $\ket{\psi}$ is an eigenvector of $U_\mathcal{Z}$, then the 4 vectors $\ket{\psi}, D_{03} \ket{\psi}, D_{30} \ket{\psi}$ and $ D_{33}\ket{\psi}$ have constant absolute values for their pairwise overlaps. Furthermore, if $\ket{\psi}$ is not in the Zauner subspace, they span a 2-dimensional subspace and therefore constitute a SIC.
\end{theorem}

\begin{proof}
If $U_\mathcal{Z}\ket{\psi}=\lambda\ket{\psi}$, then $\ket{\psi}=\lambda U_\mathcal{Z}^\dagger\ket{\psi}$. The first part of the theorem can be proven by noticing that
\begin{equation}
\bra{\psi}D_{\mathbf{p}}\ket{\psi} = \bra{\psi}\lambda^* U_\mathcal{Z} D_{\mathbf{p}}  U_\mathcal{Z}^\dagger \lambda \ket{\psi}
= \bra{\psi}D_{\mathcal{Z} \mathbf{p}}\ket{\psi} \end{equation}
and that the Zauner symplectic matrix $\mathcal{Z}$ simply permutes the three points 03, 30 and 33 in phase space according to the following table.
\begin{center}
\begin{tabular}{c | c c c }
$\mathbf{p}$ & 03 & 30 & 33\\
\hline 
$\mathcal{Z}\mathbf{p}$ & 33 & 03 & 30
\end{tabular}
\end{center}

To prove the second part of the theorem, we need to introduce a square root of the Zauner unitary. Let $\mathcal{W}$ be a symplectic matrix given by 
\begin{equation}
\mathcal{W} = \begin{pmatrix}1 & -1\\1 & 0\end{pmatrix} \hspace{8mm}
\mathcal{W}^2 = \begin{pmatrix}0 & -1\\1 & -1\end{pmatrix} = \mathcal{Z}
\end{equation}
Let $U_\mathcal{W}$ be a unitary representative of $\mathcal{W}$, with a phase chosen so that $U_\mathcal{W}^2 = U_\mathcal{Z}$. The eigenspaces of $U_\mathcal{Z}$ and $U_\mathcal{W}$ are described in Table \ref{eigenspacesHK}, where $\omega = e^{2\pi i/6}$ and $\eta = e^{2\pi i/3} =\omega^2$.
\begin{table}[ht!]
\begin{center}
\begin{tabular}{| l | c  | c |  c  | c | c |}
\hline
&\multicolumn{2}{c|}{}&\multicolumn{2}{c|}{}& \\[-3.5mm] 
Eigenspaces of $U_\mathcal{Z}$& \multicolumn{2}{c|}{$\mathcal{H}_1$ (Zauner)} &  \multicolumn{2}{c|}{$\mathcal{H}_{\eta}$} & $\mathcal{H}_{\eta^2}$  \\[0.5mm] \hline  
&\multicolumn{2}{c|}{}&\multicolumn{2}{c|}{}& \\[-3.5mm] 
Eigenvalue & \multicolumn{2}{c|}{1} & \multicolumn{2}{c|}{ $\eta$} & $\eta^2$\\ \hline
&\multicolumn{2}{c|}{}&\multicolumn{2}{c|}{}& \\[-3.5mm] 
Dimensionality & \multicolumn{2}{c|}{3}  & \multicolumn{2}{c|}{2}  & 1\\ \hline
& & & & & \\[-3.5mm] 
Eigenspaces of $U_\mathcal{W}$ & $\mathcal{K}_1$ & $\mathcal{K}_{\omega^3}$ & $\mathcal{K}_{\omega}$ & $\mathcal{K}_{\omega^4}$ & $\mathcal{K}_{\omega^2}$  \\[0.5mm] \hline  
& & & & & \\[-3.5mm] 
Eigenvalue & 1 & $\omega^3$ & $\omega$ &  $\omega^4$ & $\omega^2$\\ \hline
& & & & & \\[-3.5mm] 
Dimensionality & 2 & 1 & 1 &1 & 1\\ 
\hline
\end{tabular}
\caption{Eigenspace structure of $U_\mathcal{Z}$ and $U_\mathcal{W}$.}
\label{eigenspacesHK}
\end{center}
\end{table}

Next we define three new operators $R, S$ and $T$ as follows.
\begin{equation}\label{RST}
\begin{split}
R &= (D_{03} + D_{30} + D_{33})/\sqrt{3}\\
S &= (D_{03} + \omega^2 D_{30} + \omega^4 D_{33})/\sqrt{3}\\
T &= (D_{03} + \omega^4 D_{30} + \omega^2 D_{33})/\sqrt{3}
\end{split}
\end{equation}
Note that $\{\ket{\psi},\Dpsi{03},\Dpsi{30},\Dpsi{33}\}$ and 
$\{\ket{\psi}, R\ket{\psi}, S\ket{\psi}, T\ket{\psi}\}$ have the same linear span. We will prove that $\ket{\psi} \in \mathcal{H}_{\eta}$ implies $S\ket{\psi} = 0$ and $R\ket{\psi} = \ket{\psi}$, and that $\ket{\psi} \in \mathcal{H}_{\eta^2}$ implies $T\ket{\psi} = 0$
and $R\ket{\psi} = -\ket{\psi}$. From that one can prove that the span is 2-dimensional.

The following properties of $R, S$ and $T$ can be derived straightforwardly from their definitions:
\begin{equation}S = T^\dagger, \hspace{5mm} S^2 = T^2 = 0, \hspace{5mm} R^2 = \eye,\end{equation}
\begin{equation}ST = \eye + R, \hspace{14mm}TS = \eye - R.\end{equation}
Moreover, $ST/2$ and $TS/2$ are rank-3 projection operators, and they are orthogonal to each other. One can also verify the following  commutation relations between $R,S,T$ and $U_\mathcal{W}$
\begin{equation}U_\mathcal{W}R = R U_\mathcal{W}, \hspace{5mm} U_\mathcal{W}S = \omega^4 S U_\mathcal{W}, \hspace{5mm} U_\mathcal{W}T = \omega^2 T U_\mathcal{W}.
\end{equation}
These equations tell us how $R,S$ and $T$ permute the eigenspaces of $U_\mathcal{W}$. For example, if $\ket{\psi}$ is an eigenvector of $U_\mathcal{W}$ with eigenvalue $\omega^2$, then $U_W S \ket{\psi} = \omega^4 S U_W \ket{\psi} = S \ket{\psi}$, so $S \ket{\psi}$ is an eigenvector of  $U_\mathcal{W}$ with eigenvalue 1. The full action of $S$ and $T$ on the eigenspaces of $U_\mathcal{W}$ is described in Figure 2 ($R$ commutes with $U_\mathcal{W}$ so it simply leaves the eigenspaces invariant).
\begin{figure}[ht!]
\centering
\includegraphics[scale=0.375]{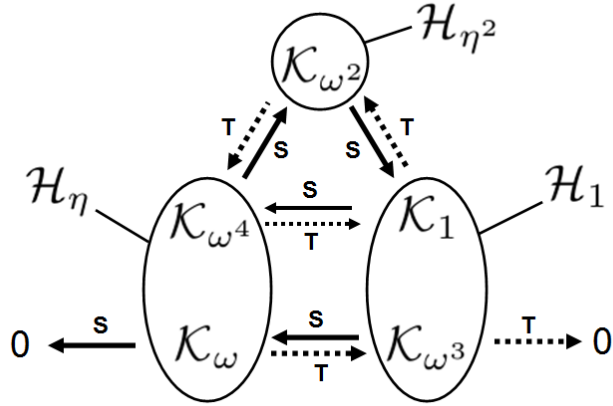}
\caption{The action of $S$ (solid arrow) and $T$ (dotted arrow) on the eigenspaces of $U_W$. $\mathcal{K}_1$ (2D) and $\mathcal{K}_{\omega^3}$ (1D) span $\mathcal{H}_{1}$ (3D). $\mathcal{K}_{\omega^4}$ (1D) and $\mathcal{K}_{\omega}$ (1D) span $\mathcal{H}_{\eta}$ (2D). $\mathcal{K}_{\omega^2}$ (1D) is $\mathcal{H}_{\eta^2}$.}
\end{figure}
Let $\ket{k_0}$, $\ket{k_1}$, $\ket{k_2}$, $\ket{k_3}$ and $\ket{k_4}$ be non-zero eigenvectors of $U_W$ belonging to the eigenspaces $\mathcal{K}_{1}$,  $\mathcal{K}_{\omega}$,  $\mathcal{K}_{\omega^2}$,  $\mathcal{K}_{\omega^3}$ and $\mathcal{K}_{\omega^4}$ respectively. Except for $\ket{k_0}$, the rest of them are unique up to a scalar, because the corresponding eigenspaces are 1-dimensional. As seen from the diagram, we have $S\ket{k_1} = 0$. We are going to prove that $S\ket{k_4}=0$, so that $S\ket{\psi} = 0$ for any $\ket{\psi} \in \mathcal{H}_{\eta}$.

Suppose otherwise, that $S\ket{k_4} \ne 0$. It has to be the case that $S\ket{k_0} = 0$ (for any choice of $\ket{k_0}$), because otherwise $S\ket{k_0}$ will be a non-zero vector in $\mathcal{K}_{\omega^4}$, which is 1-dimensional, which implies $S\ket{k_0} = \alpha \ket{k_4}$ for some non-zero $\alpha$, and therefore $S\ket{k_4} = \alpha^{-1}S^2 \ket{k_0} = 0$, contradicting the assumption that $S\ket{k_4} \ne 0$. Since $S\ket{k_4}$ is a non-zero vector in $\mathcal{K}_{\omega^2}$, which is 1-dimensional, we must also have $S\ket{k_4} = \beta \ket{k_2}$ for some non-zero $\beta$, and therefore $S\ket{k_2} = \beta^{-1}S^2 \ket{k_4}=0.$ So from the assumption that $S\ket{k_4} \ne 0$, we have obtained $0 = S\ket{k_1} = S\ket{k_2} = S\ket{k_0}$, which implies $0 = TS\ket{k_1} = TS\ket{k_2} = TS\ket{k_0}$, which means that $TS$ is orthogonal to the 4-dimensional subspace spanned by  $\mathcal{K}_1$,  $\mathcal{K}_{\omega}$ and  $\mathcal{K}_{\omega^2}$. This contradicts $TS/2$ being a rank-3 projection operator. 

Thus, we conclude that $S\ket{k_4} = 0$, and that $S\ket{\psi} = 0$ for any $\ket{\psi} \in \mathcal{H}_{\eta}$. The identity $R\ket{\psi} = \ket{\psi}$ immediately follows from $0 =TS\ket{\psi} = (\eye-R)\ket{\psi}$. Note that $T\ket{\psi} \ne 0$ (because $TS$ is orthogonal to $ST$) is a non-zero vector in $\mathcal{H}_1$. $T\ket{\psi}$ is not proportional to $\ket{\psi}$ because $T\ket{\psi}$ lies in $\mathcal{H}_{1}$ and $\ket{\psi}$ lies in $\mathcal{H}_{\eta}$. Therefore $\ket{\psi}$, $R\ket{\psi}$, $S\ket{\psi}$ and $T\ket{\psi}$ indeed fully span a 2-dimensional subspace. It's a similar argument to show that $T\ket{k_2} = 0$, which implies $T\ket{\psi}=0$ and $R\ket{\psi}=-\ket{\psi}$ for any $\ket{\psi} \in \mathcal{H}_{\eta^2}$.
\end{proof}

We can thus explain the occurrence of the $2$-dimensional SICs one finds when $N=6$.  Unfortunately, one can't explain the $3$-dimensional SICs one finds when $N=9$ in the same way.  If one uses the construction described in the proof of Theorem~4 when $N=9$ one obtains 9 vectors which do not belong to a $3$-dimensional subspace and for which the overlaps are not constant (they actually take two different values).  The construction  also fails to produce SICs  in dimensions $12$ and $15$.  We summarise the situation in Table~5, which shows the dimensionality of the subspace spanned by the $N^2$ vectors obtained using the construction of Theorem~4 as a function of the eigenspace from which $|\psi\rangle$ is chosen.

Nevertheless, the fact remains that if, in dimension $9$, instead of arbitrarily choosing a vector in $\mathcal{H}_{\eta}$ or $\mathcal{H}_{\eta^2}$ as in Theorem~4, one starts from an arbitrarily chosen fiducial vector $|\psi\rangle$ in the Zauner subspace (not necessarily a SIC fiducial), then one always does find $3$-dimensional SICs among the normal vectors to the linearly dependent sets generated by $|\psi\rangle$.  If this phenomenon repeated in higher dimensions it would open up the intriguing possibility that one might be able to get a proof of SIC existence in this  way.
 
\begin{table}
\begin{center}
\begin{tabular}{c | c  c  c  c }
\hline
                         & $N=6$ & $N=9$ & $N=12$ & $N=15$\\
\hline \\[-3.5mm]
$\mathcal{H}_1$ & 4 & 8 & 12 & 15\\
$\mathcal{H}_{\eta}$ & 2 & 7 & 8 & 15\\
$\mathcal{H}_{\eta^2}$ & 2 & 6 & 8 & 10 \\
\hline
\end{tabular}
\caption{The dimension of the span of vectors created by the subgroup generated by $D_{03}$ and $D_{30}$.}
\end{center}
\end{table}

\section{Generalisation to other symplectic unitaries}

We have been focussing so far on the Zauner unitary because, along with its conjugates, it is the symplectic unitary which appears to be relevant to the SIC existence problem.  However, there are many other symplectic unitaries which give rise to linear dependencies, as we now show.

Before proceeding further we need to take account of a subtlety which is important in even dimensions. Suppose $N$ is even, so that $\bar{N} = 2N$.  Given $\mathcal{G}\in\SL(2,\mathbb{Z}_{\bar{N}})$ let $\mathcal{G}_r\in \SL(2,\mathbb{Z}_N)$ be $\mathcal{G}$ reduced mod $N$.  We define
\begin{enumerate}
\item The $\bar{N}$-order of $\mathcal{G}$ to be its order considered as an element of $\SL(2,\mathbb{Z}_{\bar{N}})$.
\item The $N$-order of $\mathcal{G}$ to be the order of $\mathcal{G}_r$ considered as an element of $\SL(2,\mathbb{Z}_N)$.
\end{enumerate}
This distinction was not important in the case of the Zauner matrix since its $\bar{N}$-order is the same as its $N$-order.  However, the two orders are, in general, different.  Turning to $\mathbb{Z}_{\bar{N}}^2$, we say that  $\mathbf{p}$, $\mathbf{p}'\in\mathbb{Z}_{\bar{N}}^2$ are $N$-distinct if $\mathbf{p}'\neq \mathbf{p}$ mod $N$.  So as to be able to make statements which apply to both even and odd dimensions we will use the terms ``$N$-order'' and ``$N$-distinct'' in odd dimensions also, defining them in the obvious way.

Let $\mathcal{G} \in \SL(2,\mathbb{Z}_{\bar{N}})$ ($N$ even or odd).  We say that $\mathcal{G}$ generates linear dependencies in all eigenspaces if, for each eigenvector $|\psi\rangle$ of $U_\mathcal{G}$, there exist  $\mathbf{p}_1, \dots, \mathbf{p}_N$ which are $N$-distinct and  such that the vectors
\begin{equation}
D_{\mathbf{p}_1}|\psi\rangle, \dots, D_{\mathbf{p}_N}|\psi\rangle
\end{equation}
are linearly dependent.   

Continue to suppose  $\mathcal{G} \in\SL(2,\mathbb{Z}_{\bar{N}})$, and let $n$ be the $N$-order of $\mathcal{G}$.  We say that $\mathbf{p}$ is $\mathcal{G}$-full if the points $\mathbf{p}, \mathcal{G}\mathbf{p}, \dots, \mathcal{G}^{n-1}\mathbf{p}$ are $N$-distinct.  

We then have
\begin{theorem}
\label{thm:genLinDep}
Let $\mathcal{G} \in\SL(2,\mathbb{Z}_{\bar{N}})$ and let $n$ be the $N$-order of $\mathcal{G}$.  Suppose
\begin{enumerate}
\item $n>1$.
\item $N$ is a multiple of $n$.
\item \label{trCond} $\Tr(U_{\mathcal{G}})\neq 0$.
\item \label{fullCond} There exist $N$ points in $\mathbb{Z}_N^2$ which are $N$-distinct and $\mathcal{G}$-full
\end{enumerate}
Then $\mathcal{G}$ generates linear dependencies (in the sense defined above) in all eigenspaces.
\end{theorem}
\begin{proof}  We first observe that $U_{\mathcal{G}}$ has at least one eigenspace with dimension $<N/n$.  In fact, let 
\begin{equation}
U_{\mathcal{G}} = e^{i\theta} \sum_{k=0}^{n-1} \sigma^k P_k
\end{equation}
be the spectral decomposition of $U_{\mathcal{G}}$, where $\sigma=e^{\frac{2\pi i}{n}}$ and $e^{i\theta}$ is a phase.  By assumption
\begin{equation}
0 \neq \Tr(U_{\mathcal{G}}) = e^{i\theta} \sum_{k=0}^{n-1} \sigma^k \Tr(P_k)
\end{equation}
It follows that the numbers $\Tr(P_k)$ cannot all be the same.  Consequently we must have $\Tr(P_k) < N/n$ for at least one value of $k$.

It follows from condition~\ref{fullCond}  that it is possible to choose $\mathcal{G}$-full  points $\mathbf{p}_1, \dots, \mathbf{p}_{N/n}$ such that the $N$ points
\begin{equation}
\mathcal{G}^r \mathbf{p}_j, \qquad j = 1,\dots,N/n, \ r= 0, \dots, n-1 
\end{equation}
are  $N$-distinct.  Let $|\psi\rangle$ be any eigenstate of $U_{\mathcal{G}}$ and define
\begin{equation}
|\psi_{jk}\rangle = \sum_{r=0}^{n-1} \sigma^{-rk} D_{\mathcal{G}^r\mathbf{p}_j} |\psi\rangle, \qquad j = 1,\dots,N/n,\ k=0,\dots, n-1
\end{equation}
For all $k$
\begin{equation}
|\psi_{j,k}\rangle \in \mathcal{E}_k, \qquad  j = 1, \dots, N/n
\end{equation}
where $\mathcal{E}_1, \dots, \mathcal{E}_n$ are the eigenspaces $U_F$ with the appropriate ordering (note that $\mathcal{E}_k$ isn't necessarily the eigenspace corresponding to the projector $P_k$ defined above). We showed that
\begin{equation}
\dim \mathcal{E}_k < N/n
\end{equation}
for some $k$.  For that value of $k$ the vectors $|\psi_{1,k}\rangle, \dots , |\psi_{N/n,k}\rangle$ are linearly dependent, implying that the vectors
\begin{equation}
D_{\mathcal{G}^r \mathbf{p}_j} |\psi\rangle, \qquad j = 1, \dots, N/n, \ r = 0, \dots, n-1
\end{equation}
are linearly dependent.
\end{proof}

As an illustration of this theorem consider the parity matrix
\begin{equation}
\mathcal{P} = \bmt -1 & 0 \\ 0 & -1\emt
\end{equation}
We have $n=2$ and  $\Tr(U_{\mathcal{P}})=2e^{i\theta}$ (respectively $e^{i\theta}$) if $N$ is even 
(respectively odd).   Also the number of $\mathcal{G}$-full points in $\mathbb{Z}_{N}^2$ is $N^2-4$ (respectively $N^2-1$) if $N$ is even (respectively odd).  So the theorem shows that $\mathcal{P}$ generates linear dependencies in every eigenspace for all even dimensions.

It will be seen that the conditions of the theorem are not very restrictive.  There are $N^2$ points in $\mathbb{Z}_N^2$ so one would expect that the requirement that $N$ of them  be $\mathcal{G}$-full should not be difficult to satisfy. Similarly with the requirement that $\Tr(U_{\mathcal{G}})\neq 0$.  
For a large class of dimensions we can make this statement sharp:

\begin{theorem}
\label{thm:LinDepGen}
Let $N=q_1\dots q_u$ for $u$ distinct odd prime numbers $q_1, \dots, q_u$.  Then, with the sole exception of the identity matrix, every matrix $\in\SL(2,\mathbb{Z}_{\bar{N}})=\SL(2,\mathbb{Z}_{N})$  whose order is a factor of $N$  generates linear dependencies (in the sense defined above) in every eigenspace.
\end{theorem}

The proof of this theorem will  occupy us for the rest of this section.  It depends on the tensor product representation of the Clifford group which is described in Appendix B in \cite{Monomial}, and which we now briefly review.   For each $x\in\mathbb{Z}_N$ and $j$ in the range $1,\dots, u$ let $x_j$ be the unique element of $\mathbb{Z}_{q_j}$ such that $x_j = x$ mod $q_j$.  Then the Chinese Remainder Theorem~\cite{Rose} tells us that the map
\begin{equation}
x \mapsto (x_1,\dots,x_u)
\end{equation}
is an isomorphism of the ring $\mathbb{Z}_{N}$ onto the ring $\mathbb{Z}_{q_1} \times \dots \times \mathbb{Z}_{q_u}$ with inverse
\begin{equation}
(x_1,\dots, x_u) \mapsto \frac{N \kappa_1 x_1}{q_1}+\dots + \frac{N \kappa_u x_u}{q_u}
\end{equation}
where $\kappa_j$ is the multiplicative inverse of $N/q_j$ considered as an element of $\mathbb{Z}_{q_j}$.  Given 
\begin{align}
\mathbf{p} &= \bmt p_1 \\ p_2 \emt \in \mathbb{Z}_{N}^2 & \mathcal{G} & = \bmt \alpha & \beta \\ \gamma & \delta \emt \in \SL(2,\mathbb{Z}_N)
\\
\intertext{define}
\mathbf{p}_j &=\bmt p_{1j} \\ \kappa_j p_{2j} \emt  \in \mathbb{Z}_{q_j}^2
&
\mathcal{G}_j & = 
\bmt \alpha^{\vphantom{-1}}_j & \kappa_j^{-1} \beta^{\vphantom{-1}}_j 
\\ \kappa^{\vphantom{-1}}_j \gamma^{\vphantom{-1}}_j & \delta^{\vphantom{-1}}_j
\emt \in \SL(2,\mathbb{Z}_{q_j})
\label{eq:Fjdef}
\end{align}
Then define 
\begin{align}
\phi & \colon \mathbb{Z}_N^2  \to \mathbb{Z}_{q_1}^2 \otimes \dots \otimes \mathbb{Z}_{q_u}^2  &\phi(\mathbf{p}) &= \mathbf{p}_1\otimes \dots \otimes \mathbf{p}_u
\label{eq:phi}
\\
\psi &\colon \SL(2,\mathbb{Z}_N) \to \SL(2,\mathbb{Z}_{q_1}) \otimes \dots \otimes \SL(2,\mathbb{Z}_{q_u}) & \psi(\mathcal{G}) &= \mathcal{G}_1\otimes \dots \otimes \mathcal{G}_u
\label{eq:psi}
\end{align}
The map $\phi$ (respectively $\psi$) is an isomorphism of modules (respectively groups).  We have 
\begin{equation}
\phi(\mathcal{G}\mathbf{p})=\psi(\mathcal{G})\phi(\mathbf{p})
\end{equation}
for all $\mathcal{G}$, $\mathbf{p}$.
Let $\mathcal{H}_n$ be $n$-dimensional Hilbert space and let
\begin{equation}
V \colon \mathcal{H}_N \to \mathcal{H}_{q_1}\otimes \dots \otimes \mathcal{H}_{q_u}
\end{equation}
be the unitary defined by
\begin{equation}
V|x\rangle_N = |x_1\rangle_{q_1} \otimes \dots \otimes |x_u\rangle_{q_u}
\end{equation}
where $|0\rangle_n, \dots , |n-1\rangle_n$ is the standard basis in dimension $n$.  Then
\begin{align}
V D_{\mathbf{p}} V^{\dagger}& = D_{\mathbf{p}_1} \otimes \dots \otimes D_{\mathbf{p}_u}
\\
VU_{\mathcal{G}} V^{\dagger} & = U_{\mathcal{G}_1}\otimes \dots \otimes U_{\mathcal{G}_u}
\label{eq:symUTensor}
\end{align}
modulo the phase ambiguity in the definitions of $U_{\mathcal{G}}, U_{\mathcal{G}_1},\dots, U_{\mathcal{G}_u}$.
\begin{lemma}
\label{lm:FFull}
Let  $N=q_1\dots q_u$ for $u$ distinct odd prime numbers $q_1, \dots, q_u$. Let $\mathcal{G}\in \SL(2,\mathbb{Z}_N)$ be arbitrary.  Then the number of $\mathcal{G}$-full points in $\mathbb{Z}_N^2$ is $\ge N(q_1-1)\dots (q_u-1)$.
\end{lemma}
\begin{proof}
Let $K$ be the number of $\mathcal{G}$-full points in $\mathbb{Z}^2_N$ and $K_j$ the number of $\mathcal{G}_j$-full points in $\mathbb{Z}_{q_j}^2$.
Use the isomorphisms defined in Equations~(\ref{eq:phi}) and~(\ref{eq:psi}) to identify $\mathcal{G}$ with $\mathcal{G}_1\otimes \dots \otimes \mathcal{G}_u$ and $\mathbf{p}$ with $\mathbf{p}_1\otimes \dots \otimes \mathbf{p}_u$.  Let $n$ be the order of $\mathcal{G}$ and $n_j$ the order of $\mathcal{G}_j$.  Then
\begin{equation}
n = \LCM(n_1,\dots,n_u)
\end{equation}
Suppose $\mathbf{p}_j$ is $\mathcal{G}_j$-full in $\mathbb{Z}_{q_j}^2$ for each $j$.  Then
\begin{equation}
\mathcal{G}^r \mathbf{p} = \mathbf{p}
\end{equation}
implies $r$ is a multiple of $n_j$ for each $j$, which in turn implies $r$ is a multiple of $n$.  So $\mathbf{p}$ is $\mathcal{G}$-full in $\mathbb{Z}^2_{q_j}$.  So $K \ge K_1 \dots K_u$. The problem thus reduces to showing that $K_j\ge q_j(q_j-1)$ for all $j$.

To see this we use the analysis in~\cite{Appleby09}.  Let $t_j = \Tr(\mathcal{G}_j)$ and let $Q_j$ (respectively $N_j$) be the set of quadratic residues (respectively non-residues) in $\mathbb{Z}_{q_j}$.  If $t_j^2-4\in Q_j$ then $\mathcal{G}_j$ is conjugate to a diagonal matrix
\begin{equation}
\tilde{\mathcal{G}}_j = \bmt \alpha & 0 \\ 0 & \alpha^{-1}\emt
\end{equation}
If $\mathbf{p}_j$ is non-zero then
\begin{equation}
\tilde{\mathcal{G}}_j^r \mathbf{p_j} = \mathbf{p_j}
\end{equation}
if and only if $r$ is a multiple of $n_j$.  So $K_j=q_j^2-1$.  If $t_j^2-4\in N_j$ we can diagonalize $\mathcal{G}_j$ over  the extension field $\mathbb{F}_{q^2_j}$.  So we again have that $\mathcal{G}_j$ is conjugate to a matrix of the form
\begin{equation}
\tilde{\mathcal{G}}_j = \bmt \alpha & 0 \\ 0 & \alpha^{-1}\emt
\end{equation} 
where $\alpha$ now belongs to $\mathbb{F}_{q^2_j}$ instead of $\mathbb{Z}_{q_j}$.  We again deduce that $K_j = q_j^2-1$.  Finally, if $t_j = \pm 2$ we have that $\mathcal{G}_j$ is conjugate to the matrix 
\begin{equation}
\tilde{\mathcal{G}}_j = \bmt \pm 1 & \beta \\ 0 & \pm 1 \emt
\end{equation}
for some $\beta$.  If $\beta=0$ then $K_j = q_j^2-1$; otherwise $K_j = q_j^2-q_j$.
\end{proof}
\begin{lemma}
\label{lm:Espce}
Let $N=q_1\dots q_u$ for $u$ distinct odd prime numbers $q_1, \dots, q_u$.  Then $|\Tr(U_{\mathcal{G}})|\ge 1$ for all $\mathcal{G}\in \SL(2,\mathbb{Z}_{N})$.  
\end{lemma}
\begin{proof}
In view of Equation~(\ref{eq:symUTensor})  it is enough to show that
\begin{equation}
|\Tr(U_{\mathcal{G}_j})|\ge 1
\end{equation}
for all $j$.    To see this let
\begin{equation}
\mathcal{G}_j = \bmt \alpha & \beta \\ \gamma & \delta \emt
\end{equation}
We have~\cite{Appleby09}
\begin{equation}
U_{\mathcal{G}_j} =
\begin{cases}
\frac{e^{i\theta}}{\sqrt{q_j}} \sum_{x,y=0}^{q_j-1} \tau_j^{\beta^{-1}(\delta x^2-2 x y + \alpha y^2} |x\rangle \langle y |
\qquad & \beta \neq 0
\\
e^{i\theta} \sum_{x=0}^{q_j-1} \tau_j^{\alpha \gamma x^2} |\alpha x\rangle \langle x |
\qquad & \beta = 0
\end{cases}
\end{equation}
where $\tau_j = e^{\frac{2\pi i}{q_j}}$ and $e^{i\theta}$ is a phase.  So
\begin{equation}
\Tr(U_{\mathcal{G}_j}) = 
\begin{cases}
\frac{e^{i\theta}}{\sqrt{q_j}} \sum_{x=0}^{q_j-1} \omega_j^{2^{-1}\beta^{-1}(t-2)x^2} \qquad &\beta \neq 0
\\
e^{i\theta} \qquad & \beta = 0,\  \alpha \neq 1
\\
e^{i\theta}\sum_{x=0}^{q_j-1} \omega_j^{2^{-1}\gamma x^2} \qquad &\beta=0,\ \alpha = 1
\end{cases}
\end{equation}
where $t=\Tr(\mathcal{G}_j)$ and $\omega_j = e^{\frac{2\pi i}{q_j}}$.  Performing the Gauss sums~\cite{Rose} we find
\begin{align}
\Tr(U_{\mathcal{G}_j}) &=
\begin{cases}
\eta e^{i\theta} \lsym{2\beta(t-2)}{q_j} \qquad & \beta\neq 0 , \ t\neq 2
\\
\eta e^{i\theta} \sqrt{q_j} \qquad &\beta\neq 0 , \ t = 2
\\
e^{i\theta} \qquad & \beta = 0 , \ \alpha\neq 1
\\
\eta e^{i\theta} \sqrt{q_j} \lsym{2\gamma}{q_j} \qquad & \beta = 0 , \ \alpha = 1, \ \gamma \neq 0
\\
e^{i\theta}q_j \qquad & \beta = 0,\ \alpha = 1,\  \gamma = 0
\end{cases}
\end{align}
where 
\begin{equation}
\eta = \begin{cases} 1 \qquad & q_j = 1 \mod 4 \\ i \qquad & q_j = -1 \mod 4 \end{cases}
\end{equation}
where $\lsym{x}{q_j}$ is the Legendre symbol, and where we have used the multiplicative property of the Legendre symbol together with the fact that $\lsym{x^{-1}}{q_j}=\lsym{x}{q_j}$ for all non-zero $x$.  We see that in every case
\begin{equation}
|\Tr(U_{\mathcal{G}_j})| = \text{$1$, $\sqrt{q_j}$ or $q_j$}
\end{equation}
\end{proof}
\begin{proof}[Proof of Theorem~\ref{thm:LinDepGen}] 
Immediate consequence of Theorem~\ref{thm:genLinDep} and Lemmas~\ref{lm:FFull}, \ref{lm:Espce}.
\end{proof}
 It would be interesting to see if one can generalise Theorem~\ref{thm:LinDepGen} to the case of even dimensions, and odd dimensions for which the prime factorization contains primes occurring with multiplicity $>1$.

\section{Linear dependencies in the monomial representation}

In the motivating three-dimensional example there was an interesting interplay between linear dependencies in WH orbits on the one hand, and preferred orthonormal bases on the other. It is not difficult to find fiducial vectors giving rise to a WH orbit with $N$ linearly dependent sets of $N$ vectors each, in such a way that their normal vectors form an orthonormal basis: if the generator $Z$ is diagonal, any fiducial vector with a zero entry will give rise to such a linearly dependent set under repeated action with $Z$. Acting with $X$ on the fiducial will permute its entries, and one ends up with $N$ linearly dependent sets in the orbit. 

The situation is more interesting if the dimension is a square number, $N = n^2$. For this case it has recently been shown that there exists a choice of basis in which the entire Clifford group is represented by monomial matrices, i.e. by permutation matrices in which the unit entries have been replaced by phase factors \cite{Monomial}. This representation stems from the observation that in square dimensions 
\begin{equation} 
[X^n,Z^n] = 0 
\end{equation}
Thus $X^n$ and $Z^n$ can be simultaneously diagonalised, and the resulting basis turns out to have the above property. Denote the basis vectors by $|r,s\rangle $, where $r,s$ are integers modulo $n$, and let us refer to the basis as the monomial basis. The monomial representation of the Clifford group is given by 
\begin{equation} 
X|r,s\rangle = \left\{ \begin{array}{ll} |r,s+1\rangle 
& \mbox{if} \ s+1 \neq 0 \ \mbox{mod} \ n \\
\ \\
q^r|r, 0\rangle & \mbox{if} \ s + 1 = 0 \ \mbox{mod} \ n \end{array} \right. 
\label{WHrep} 
\end{equation} 
\begin{eqnarray} 
Z|r,s\rangle = \omega^s |r-1,s\rangle  
\end{eqnarray}
where $q = e^{\frac{2 \pi i}{n}}$.
\begin{equation}
 U_G|r,s\rangle = e^{i\theta}\tau^{\beta^{-1}
(\delta s'^2 -2ss' + \alpha s^2)}|\delta r - \gamma s + 
m \gamma \delta, - \beta r + \alpha s + m\alpha \beta \rangle 
 \label{therep} 
\end{equation}
where 
\begin{equation}
 s' = - \beta r + \alpha s + m\alpha 
 \hspace{8mm} m = \left\{ \begin{array}{lll} 0 & \ & n \ \mbox{is odd} \\ 
\frac{n}{2} & \ & n \ \mbox{is even}  \end{array} \right. 
\end{equation}
Note that Equation (\ref{therep}) assumes $\beta$ is relatively prime to $\bar{N}$, otherwise one has to use the decomposition given in Equation (\ref{eq:primeDecomp}) of Section 2. Evidently the $N$ orthogonal rays that make up the basis is in itself an orbit of the full Clifford group. Moreover, one sees by inspection that the vector $|0,0\rangle $ is left invariant by the full symplectic group whenever $N$ is odd. In this sense the monomial basis is a distant cousin of the `special SIC' in three dimensions---or perhaps of the MUB in three dimensions.

Now suppose that $3|n$, and consider the Zauner matrix 
\begin{equation} 
{\cal Z} = \left( \begin{array}{rr} 0 & - 1 \\ 1 & - 1 \end{array} 
\right) \ 
\end{equation}
If we scrutinise the explicit expression (\ref{therep}) we find that there are three non-zero diagonal elements in the unitary operator $U_{\cal Z}$. If we adjust the phase $e^{i\theta}$ according to Zauner's recipe one of them equals $\eta$, and the other two equal $1$. It follows that any vector in the ${\cal H}_1$ eigenspace must have one component equal to 0. Since the WH group contains $N$ diagonal elements (namely the elements of the subgroup generated by $X^n$ and $Z^n$) this means that the WH orbit generated from a fiducial vector in ${\cal H}_1$ contains $N$ vectors confined to a subspace orthogonal to one of the vectors in the monomial basis. Since the entire WH group is given by monomial matrices the orbit can be divided into $N$ such linearly dependent sets, spanning $N$ subspaces each orthogonal to a vector in the monomial basis. In particular this applies to the known SICs in these dimensions.  

A similar argument applies if $2|n$. We consider the symplectic matrix 
\begin{equation} 
{\cal F} = \left( \begin{array}{rr} 0 & 1 \\ - 1 & 0 \end{array} 
\right) 
\end{equation}   
After adjusting the arbitrary phase so that $U_{\cal F}^4 = {\bf 1}$ and so that the largest eigenspace corresponds to the eigenvalue 1, we find that $U_{\cal F}$ has two diagonal non-zero elements, one equal to 1 and the other equal to $\omega^{N/4}$. Following the same logic as above we see that every WH orbit that includes a fiducial invariant under $U_{\cal F}$ must be divided into $N$ linearly dependent sets of $N$ vectors orthogonal to one of the vectors in the monomial basis. 

Variants of this argument can be used also if the dimension is $N = kn^2$, where $k$ is any integer. In this case the Hilbert space splits into a direct sum of $n^2$ copies of a $k$-dimensional Hilbert space, and there exists a representation of the Clifford group in which every group element acts within and permutes these subspaces \cite{monomial2}. This representation is $k$-nomial, in the sense that it is given by a permutation matrix of size $n^2$ in which the unit entries have been replaced by matrices of size $k$. For instance, one finds in dimension eight that any vector in the smallest eigenspace of the Zauner unitary has two zero entries, and again there will be linear dependencies in a WH orbit arising from a fiducial in this eigenspace. A SIC fiducial does in fact exist in this eigenspace \cite{Sco10}.

\section{Conclusion}

The Hesse configuration singles out `special SICs' in dimension 3 by the higher number of linear dependencies among their vectors. We searched for similar patterns in higher dimensions through exhaustive numerical searches for sets of $N$ linearly dependent vectors in SICs in dimensions $N=4$ to 8 and partial searches in dimensions $N=9$ and 12. In dimension 8, the SICs with fiducials in the smallest subspace of the Zauner unitary, $\mathcal{H}_{\eta^2}$, also stand out for exhibiting a higher number of linear dependencies than other WH orbits with a fiducial vector from the same subspace. This could extend the dimension 3 connection between elliptic curves and SICs (which gives rise to the Hesse configuration) to dimension 8 and, as SIC fiducials always lie in $\mathcal{H}_{\eta^2}$ for dimensions equal to 8 mod 9, possibly to an infinite family of dimensions. 

In the dimensions we looked at, sets of $N$ SIC vectors are usually linearly independent. We proved that it is always possible to find dependencies in SICs for dimensions divisible by 3 (where the SIC fiducials sit in $\mathcal{H}_{1}$) and for dimensions equal to 8 mod 9 (where the SIC fiducials sit in $\mathcal{H}_{\eta^2}$). However, this dependency structure in dimensions divisible by 3 also occurs for any WH orbit when the fiducial vector is taken from $\mathcal{H}_{1}$ and so is not a property uniquely tied to SICs. We also proved that linear dependencies will arise for WH orbits whose fiducial vector is in $\mathcal{H}_{1}$ for dimensions $N=3k$; WH orbits whose fiducial vector is in $\mathcal{H}_{\eta}$ for dimensions $N=3k$ and $N=3k+1$; and WH orbits whose fiducial vector is in $\mathcal{H}_{\eta^2}$ for all dimensions. These theorems do not account for all the linear dependencies; our numerical results often revealed a higher number of linearly dependent sets of vectors than the theorems predict.

Looking in detail at linear dependencies in dimension 6 revealed some surprising structure. Firstly, the normal vectors (i.e. the vectors orthogonal to each 5-dimensional subspace spanned by 6 linearly dependent vectors) collect into orthogonal sets of 4 when the HW orbit is a SIC. When the orbit is not a SIC, these normal vectors lose their orthogonality pattern. Secondly, the normal vectors of dependencies from both SICs and non-SICs form 30 2-dimensional SICs. This repeats in dimension 9, where some of the normal vectors form 3-dimensional SICs. We did not find 4-dimensional SICs from normal vectors in dimension 12, but our search was not exhaustive and we cannot rule out their existence either.

We also proved a number of theorems for linear dependencies using unitaries other than the Zauner unitary. Our results represent the beginning of a theory for when linear dependencies occur in WH orbits. This is of interest in itself. Also, just possibly, they may contain an important clue for the SIC existence problem.

\section*{Acknowledgements}

IB and KB thank Markus Grassl for taking an interest in this problem when we visited Singapore, and especially for drawing our attention to the linear dependencies in the ``odd'' SIC in eight dimensions. We also thank a referee for constructive comments. HBD was supported by the Natural Sciences and Engineering Research Council of Canada and by the U.S. Office of Naval Research (Grant No. N00014-09-1-0247). IB was supported by the Swedish Research Council under contract VR 621-2010-4060. DMA was supported in part by the U.S. Office of Naval Research (Grant No.\ N00014-09-1-0247) and by the John Templeton Foundation. Research at Perimeter Institute is supported by the Government of Canada through Industry Canada and by the Province of Ontario through the Ministry of Research \& Innovation.

\end{document}